\documentclass[11pt]{article}

\usepackage{graphicx,amssymb,amsmath}
\usepackage{amsfonts}
\usepackage{booktabs}
\usepackage{cite}
\usepackage{color}
\usepackage{enumerate}
\usepackage{float}
\usepackage{hyperref}
\usepackage{mathrsfs}
\usepackage{subfig}
\usepackage{tabularx}
\usepackage{authblk}

\usepackage[in]{fullpage}
\usepackage[top=0.8in, bottom=0.9in, left=0.9in, right=0.9in]{geometry}

\def\calA{\mathcal{A}}
\def\calH{\mathcal{H}}
\def\calL{\mathcal{L}}

\newtheorem{lemma}{Lemma}
\newtheorem{theorem}{Theorem}
\newenvironment{proof}{\noindent {\textbf{Proof:}}\rm}{\hfill $\Box$ \rm\bigskip}

\title{Constructing Many Faces in Arrangements of Lines and Segments\thanks{This research was supported in part by NSF under Grant CCF-2005323.}}

\author{
Haitao Wang
}
\affil{Department of Computer Science \\
Utah State University, Logan, UT 84322, USA
\\ {\tt haitao.wang@usu.edu}}

\begin{document}

\pagestyle{plain}
\pagenumbering{arabic}
\setcounter{page}{1}
\date{}

\thispagestyle{empty}
\maketitle

\vspace{-0.35in}
\begin{abstract}
We present new algorithms for computing many faces in arrangements of lines and segments. Given a set $S$ of $n$ lines (resp., segments) and a set $P$ of $m$ points in the plane, the problem is to compute the faces of the arrangements of $S$ that contain at least one point of $P$.

For the line case, we give a deterministic algorithm of $O(m^{2/3}n^{2/3}\log^{2/3} (n/\sqrt{m})+(m+n)\log n)$ time. This improves the previously best deterministic algorithm [Agarwal, 1990] by a factor of $\log^{2.22}n$ and improves the previously best randomized algorithm [Agarwal, Matou\v{s}ek, and Schwarzkopf, 1998] by a factor of $\log^{1/3}n$ in certain cases (e.g., when $m=\Theta(n)$).

For the segment case, we present a deterministic algorithm of $O(n^{2/3}m^{2/3}\log n+\tau(n\alpha^2(n)+n\log m+m)\log n)$ time, where $\tau=\min\{\log m,\log (n/\sqrt{m})\}$ and $\alpha(n)$ is the inverse Ackermann function. This improves the  previously best deterministic algorithm [Agarwal, 1990] by a factor of $\log^{2.11}n$ and improves the previously best randomized algorithm [Agarwal, Matou\v{s}ek, and Schwarzkopf, 1998] by a factor of $\log n$ in certain cases (e.g., when $m=\Theta(n)$).
We also give a randomized algorithm of $O(m^{2/3}K^{1/3}\log n+\tau(n\alpha(n)+n\log m+m)\log n\log K)$ expected time, where $K$ is the number of intersections of all segments of $S$.

In addition, we consider the query version of the problem, that is, preprocess $S$ to compute the face of the arrangement of $S$ that contains any query point. We present new results that improve the previous work for both the line and the segment cases.
\end{abstract}


\section{Introduction}
\label{sec:intro}

We consider the problem of computing many faces in arrangements of lines and segments. Given a set $S$ of $n$ lines (resp., segments) and a set $P$ of $m$ points in the plane, the problem is to compute the faces of the arrangement of $S$ that contain at least one point of $P$. Note that faces in an arrangement of lines are convex, but they may not even be simply connected in an arrangement of segments.
These are classical problems in computational geometry and have been studied in the literature. There has been no progress on these problems for more than two decades. In this paper, we present new algorithms that improve the previous work.

\paragraph{The line case.}
For the line case where $S$ consists of $n$ lines, it has been proved that the combinatorial
complexity of all faces of the arrangement that contain at least one
point of $P$ is bounded by
$O(m^{2/3}n^{2/3}+n)$~\cite{ref:ClarksonCo90} (which matches the
$\Omega(m^{2/3}n^{2/3}+n)$ lower bound~\cite{ref:EdelsbrunnerOn86}),
as well as bounded by $O(n\sqrt{m})$ and
$O(n+m\sqrt{n})$~\cite{ref:EdelsbrunnerOn86}. To compute these faces,
a straightforward approach is to first construct the arrangement of
$S$ and then find the faces using point
locations~\cite{ref:EdelsbrunnerOp86,ref:KirkpatrickOp83}. This takes
$O(n^2+m\log n)$ time in total. Edelsbrunner, Guibas, and
Sharir~\cite{ref:EdelsbrunnerTh90} gave a randomized algorithm of
$O(m^{2/3-\delta}n^{2/3+2\delta}\log n+n\log n\log m)$ expected time
for any $\delta>0$. Later Agarwal~\cite{ref:AgarwalPa902} presented an
	improved deterministic algorithm of $O(m^{2/3}n^{2/3}\log^{5/3}n
	\log^{1.11}(m/\sqrt{n})+(m+n)\log n)$ time; Agarwal, Matou\v{s}ek,
	and Schwarzkopf~\cite{ref:AgarwalCo98} proposed a randomized
	algorithm of $O(m^{2/3}n^{2/3}\log (n/\sqrt{m})+(m+n)\log n)$
	expected time.
On the other hand, the problem has a lower bound of $\Omega(m^{2/3}n^{2/3}+n\log n + m)$ time due to the above $\Omega(m^{2/3}n^{2/3}+n)$ lower bound~\cite{ref:EdelsbrunnerOn86} on the combinatorial complexity of all these faces and also because computing a single face in line arrangements requires $\Omega(n\log n)$ time.

We propose a new deterministic algorithm of $O(m^{2/3}n^{2/3}\log^{2/3} (n/\sqrt{m})+(m+n)\log n)$ time. In  certain cases (e.g., when $m=\Theta(n)$), our result improves the deterministic algorithm of~\cite{ref:AgarwalPa902} by a factor of $\log^{2.22}n$ and improves the randomized algorithm of~\cite{ref:AgarwalCo98} by a factor of $\log^{1/3}n$.

Our algorithm follows the framework of Agarwal~\cite{ref:AgarwalPa902}, which uses a cutting of $S$ to divide the problem into a collection of subproblems. To solve each subproblem,
Agarwal~\cite{ref:AgarwalPa902} derived another algorithm of $O(n\log n+ m\sqrt{n}\log^2 n)$ time. Our main contribution is a more efficient algorithm of $O(n\log n+ m\sqrt{n\log n})$ time. Using our new algorithm to solve the subproblems induced by the cutting, the asserted result can be achieved. The algorithm of~\cite{ref:AgarwalCo98} also follows a similar framework, but it uses the random sampling technique~\cite{ref:ClarksonAp89} instead of the cutting to divide the problem, and a randomized algorithm of $O(n\log n+ m\sqrt{n}\log n)$ expected time was proposed in~\cite{ref:AgarwalCo98} to solve each subproblem.
In particular, our algorithm runs in $O(n\log n)$ time for $m= O(\sqrt{n\log n})$, which matches the $\Omega(n\log n)$ lower bound for computing a single face (for comparison, the randomized algorithm of~\cite{ref:AgarwalCo98} runs in $O(n\log n)$ expected time for $m=O(\sqrt{n})$).

\paragraph{The segment case.}
For the segment case where $S$ consists of $n$ line segments, it is known that the combinatorial complexity of all faces of the arrangement that contain at least one point of $P$ is upper bounded by $O(m^{2/3}n^{2/3}+n\alpha(n)+n\log m)$~\cite{ref:AronovTh92} and $O(\sqrt{m}n\alpha(n))$~\cite{ref:EdelsbrunnerAr92}, as well as lower bounded by $\Omega(m^{2/3}n^{2/3}+n\alpha(n))$~\cite{ref:EdelsbrunnerTh90}, where $\alpha(n)$ is the inverse Ackermann function. To compute these faces, as in the line case, a straightforward approach is to first construct the arrangement of $S$ and then find the faces using point locations~\cite{ref:EdelsbrunnerOp86,ref:KirkpatrickOp83}. This takes $O(n^2+m\log n)$ time in the worst case (more precisely, the arrangement can be constructed in $O(n\log n + K)$ time~\cite{ref:BalabanAn95,ref:ChazelleAn92Edelsbrunner} or by simpler randomized algorithms of the same expected time~\cite{ref:ClarksonAp89,ref:ChazelleCo93,ref:MulmuleyA90}; throughout the paper, we use $K$ to denote the number of intersections of all segments of $S$).

Edelsbrunner, Guibas, and Sharir~\cite{ref:EdelsbrunnerTh90} gave a randomized algorithm of $O(m^{2/3-\delta}n^{2/3+2\delta}\log n+n\alpha(n)\log^2 n\log m)$ expected time for any $\delta>0$. Agarwal~\cite{ref:AgarwalPa902} presented an improved deterministic algorithm of $O(m^{2/3}n^{2/3}\log n \log^{2.11}(n/\sqrt{m})+n\log^3 n+ m\log n)$ time. Agarwal, Matou\v{s}ek, and Schwarzkopf~\cite{ref:AgarwalCo98} derived a randomized algorithm of
$O(n^{2/3}m^{2/3}\log^2 (K/m)+(n\alpha(n)+n\log m+m)\log n)$ expected time and another algorithm of
$O(m^{2/3}K^{1/3}\log^2 (K/m)+(n\alpha(n)+n\log m+m)\log n)$ expected time\footnote{It appears that their time analysis~\cite{ref:AgarwalCo98} is based on the assumption that $K$ is known. If $K$ is not known, their algorithm could achieve $O(m^{2/3}K^{1/3}\log^2 (K/m)+(m+n\log m+n\alpha(n))\log n\log K)$ expected time by the standard trick of ``guessing'', which is also used in this paper.}.
On the other hand, the lower bound $\Omega(m^{2/3}n^{2/3}+n\log n + m)$ for the line case is also applicable here (and we are not aware of any better lower bound). Note that computing a single face in an arrangement of segments can be done in $O(n\alpha(n)\log n)$ expected time by a randomized algorithm~\cite{ref:ChazelleCo93} or in $O(n\alpha^2(n)\log n)$ time by a deterministic algorithm~\cite{ref:AmatoCo95} (which improve the previous $O(n\log^2 n)$ time algorithm~\cite{ref:MitchellOn90} and $O(n\alpha(n)\log^2 n)$ time algorithm~\cite{ref:EdelsbrunnerTh90};
but computing the upper envelope can be done faster in $O(n\log n)$ time~\cite{ref:HershbergerFi89}).

We propose a new deterministic algorithm of $O(n^{2/3}m^{2/3}\log n+\tau(n\alpha^2(n)+n\log m+m)\log n)$ time, where $\tau=\min\{\log m,\log (n/\sqrt{m})\}$. In certain cases (e.g., when $m=\Theta(n)$ and $K=\Theta(n^2)$), our result improves the deterministic algorithm of~\cite{ref:AgarwalPa902} by a factor of $\log^{2.11}n$ and improves the randomized algorithm of~\cite{ref:AgarwalCo98} by a factor of $\log n$.
In particular, the algorithm runs in $O(n\alpha^2(n)\log n)$ time for $m= O(1)$, which matches the time for computing a single face~\cite{ref:AmatoCo95}, and runs in $O(m\log n)$ time for $m= \Theta(n^2)$, which matches the performance of the above straightforward approach. Our algorithm uses a different approach than the previous work~\cite{ref:AgarwalPa902,ref:AgarwalCo98}. In particular,
our above algorithm for the line case is utilized as a subroutine.

If $K=o(n^2)$, we further obtain a faster randomized algorithm of $O(m^{2/3}K^{1/3}\log n+\tau(n\alpha(n)+n\log m+m)\log n\log K)$ expected time, where $\tau=\min\{\log m,\log (n/\sqrt{m})\}$. This improves the result of~\cite{ref:AgarwalCo98} by a factor of $\log n$ for relative large values of $K$, e.g., when $m=\Theta(n)$ and $K=\Omega(n^{1+\epsilon})$ for any constant $\epsilon \in (0,1]$.
Our above deterministic algorithm (with one component replaced by a faster randomized counterpart) is utilized as a subroutine.

\paragraph{The face query problem.}
We also consider a related face query problem in which we wish to preprocess $S$ so that given a query point $p$, the face of the arrangement containing $p$ can be computed efficiently.

For the line case, inspired by our techniques for computing many faces and utilizing the randomized optimal partition tree of Chan~\cite{ref:ChanOp12}, we construct a data structure of $O(n\log n)$ space in $O(n\log n)$ randomized time, so that the face $F_p(S)$ of the arrangement of $S$ that contains a query point $p$ can be computed and the query time is bounded by $O(\sqrt{n}\log n)$ with high probability. More specifically, the query algorithm returns a binary search tree representing the face $F_p(S)$ so that standard binary-search-based queries on $F_p(S)$ can be handled in $O(\log n)$ time each, and $F_p(S)$ can be output explicitly in $O(|F_p(S)|)$ time. Previously, Edelsbrunner, Guibas, Hershberger, Seidel, Sharir, Snoeyink, and Welzl~\cite{ref:EdelsbrunnerIm89} built a data structure of $O(n\log n)$ space in $O(n^{3/2}\log^2n)$ randomized time, and the query time is bounded by $O(\sqrt{n}\log^5 n)$ with high probability, which is further reduced to $O(\sqrt{n}\log^2 n)$ in~\cite{ref:GuibasCo91} using compact interval trees.
Thus, our result improves their preprocessing time by a factor of $\sqrt{n}\log n$ and improves their query time by a factor of $\log n$. We further obtain a tradeoff between the storage and the query time. For any value $r<n/\log^{\omega(1)} n$, we construct a data structure of $O(n\log n+nr)$ space in $O(n\log n+nr)$ randomized time, and the query time is bounded by $O(\sqrt{n/r}\log n)$ with high probability.

For the segment case, the authors~\cite{ref:EdelsbrunnerIm89} also gave
a data structure for the face query problem with the following performance: the preprocessing takes $\widetilde{O}(n^{5/3})$ randomized time, the space is $\widetilde{O}(n^{4/3})$, and the query time is bounded by $\widetilde{O}(n^{1/3})+O(\kappa)$ with high probability, where the notation $\widetilde{O}$ hides a polylogarithmic factor and $\kappa$ is the size of the query face (note that $\kappa$ can be $\Theta(n\alpha(n))$ in the worst case~\cite{ref:EdelsbrunnerTh90} and the face may not be simply connected). Their preprocessing algorithm uses the query algorithm for the line case as a subroutine. If we follow their algorithmic scheme but instead use our new query algorithm for the line case
as the subroutine, then the preprocessing time can be reduced to $\widetilde{O}(n^{4/3})$, while the space is still $\widetilde{O}(n^{4/3})$ and the query time is still bounded by $\widetilde{O}(n^{1/3})+O(\kappa)$ with high probability.

\paragraph{Outline.}
The rest of the paper is organized as follows. We define notation and introduce some concepts in Section~\ref{sec:pre}. Our algorithms for computing many faces are described in Sections~\ref{sec:line} and~\ref{sec:segment}. The query problem is discussed in Section~\ref{sec:query}.

\section{Preliminaries}
\label{sec:pre}

We define some notation that is applicable to both the line and segment cases.
Let $S$ be a set of $n$ line segments (a line is considered  a
special line segment) and let $P$ be a set of $m$ points in the plane.
For a subset $S'\subseteq S$, we use $\calA(S')$ to denote the
arrangement of $S'$. For any point $p\in P$, we use $F_p(S')$ to
denote the face of $\calA(S')$ that contains $p$.
A face of $\calA(S')$ is {\em nonempty} if it contains a point of $P$.
Hence, the problem of computing many faces is to compute all
nonempty cells of $\calA(S)$. Note that if a nonempty face contains
more than one point of $P$, then we need to output the face only once.

For any compact region $A$ and a set $Q$ of points in the plane, we often use $Q(A)$ to denote
the subset of $Q$ in $A$, i.e., $Q(A)=Q\cap A$.



\paragraph{Cuttings.}
Let $H$ be a set of $n$ lines in the plane. For a compact region $A$ in the plane, we use $H_A$ to denote the subset of lines of $H$ that intersect the interior of $A$ (we also say that these lines {\em cross} $A$).
A {\em cutting} for $H$ is a collection $\Xi$ of closed cells (each of which is a triangle) with disjoint interiors, which together cover the entire plane~\cite{ref:ChazelleCu93,ref:MatousekRa93}.
The {\em size} of $\Xi$ is the number of cells in $\Xi$.
For a parameter $r$ with $1\leq r\leq n$, a {\em $(1/r)$-cutting} for $H$ is a cutting $\Xi$ satisfying $|H_{\sigma}|\leq n/r$ for every cell $\sigma\in \Xi$.



A cutting $\Xi'$ {\em $c$-refines} another cutting $\Xi$ if every cell of $\Xi'$ is contained in a
single cell of $\Xi$ and every cell of $\Xi$ contains at most $c$
cells of $\Xi'$. A {\em hierarchical $(1/r)$-cutting} (with two
constants $c$ and $\rho$) is a sequence
of cuttings $\Xi_0,\Xi_1,\ldots,\Xi_k$ with the following properties.
$\Xi_0$ is the entire plane. For each $1\leq i\leq k$, $\Xi_i$ is a
$(1/\rho^i)$-cutting of size $O(\rho^{2i})$ which $c$-refines
$\Xi_{i-1}$. In order to make
$\Xi_k$ a $(1/r)$-cutting, we set $k=\Theta(\log r)$ so that
$\rho^{k-1}<r\leq \rho^k$. Hence, the size of $\Xi_k$ is $O(r^2)$. If a cell $\sigma\in
\Xi_{i-1}$ contains a cell $\sigma'\in \Xi_i$, we say that $\sigma$ is
the {\em parent} of $\sigma'$ and $\sigma'$ is a {\em child} of
$\sigma$. As such, one could view $\Xi$ as a tree structure in which each node corresponds to a cell $\sigma\in \Xi_i$, $0\leq i\leq k$.

For any $1\leq r\leq n$, a hierarchical $(1/r)$-cutting of size
$O(r^2)$ for $H$ (together with the sets $H_{\sigma}$ for every cell
$\sigma$ of $\Xi_i$ for all $i=0,1,\ldots,k$) can be computed in
$O(nr)$ time by Chazelle's algorithm~\cite{ref:ChazelleCu93}.

\section{Computing many cells in arrangements of lines}
\label{sec:line}

In this section, we consider the line case for computing many cells. Let $S$ be a set of $n$ lines and $P$ be a set of $m$ points in the plane. Our goal is to compute the nonempty cells of the arrangement $\calA(S)$.
For ease of exposition, we make a general position assumption that no line of $S$ is vertical, no three lines of $S$ are concurrent, and no point of $P$ lies on a line of $S$. Degenerate cases can be handled by standard techniques~\cite{ref:EdelsbrunnerSi90}. Under the assumption, each point of $P$ is in the interior of a face of $\calA(S)$.

First of all, if $m\geq n^2/2$, then the problem can be solved in $O(m\log n)$ time using the straightforward algorithm mentioned in Section~\ref{sec:intro} (i.e., first compute $\calA(S)$ and then find the nonempty cells using point location). In what follows, we assume that $m<n^2/2$.
Our algorithm follows the high-level scheme of Agarwal~\cite{ref:AgarwalPa902} by using a cutting of $S$ to divide the problem into many subproblems. The difference is that we develop an improved algorithm for solving each subproblem.
In the following, we first present an algorithm of $O(n\log n+m\sqrt{n\log n})$ time in Section~\ref{sec:first}, and then use it to solve each subproblem and thus obtain our main algorithm with the asserted time in Section~\ref{sec:second}.

\subsection{The first algorithm}
\label{sec:first}

We say that $S$ and $P$ are in the primal plane and we consider the problem in the dual plane. Let $S^*$ be the set of dual points of $S$ and let $P^*$ be the set of dual lines of $P$.\footnote{We use the following duality~\cite{ref:deBergCo08}: A point $(a,b)$ in the primal plane is dual to the line $y=ax-b$ in the dual plane; a line $y=cx+d$ in the primal plane is dual to the point $(c,-d)$ in the dual plane.} Consider a point $p\in P$ and the face $F_p(S)$ of $\calA(S)$ that contains $p$. In the dual plane, the dual line $p^*$ of $p$ partitions $S^*$ into two subsets and the portions of the convex hulls of the two subsets between their inner common tangents are dual to the face $F_p(S)$~\cite{ref:AgarwalPa902,ref:EdelsbrunnerIm89}; e.g., see Fig~\ref{fig:dual}.

\begin{figure}[t]
\begin{minipage}[t]{\textwidth}
\begin{center}
\includegraphics[height=1.6in]{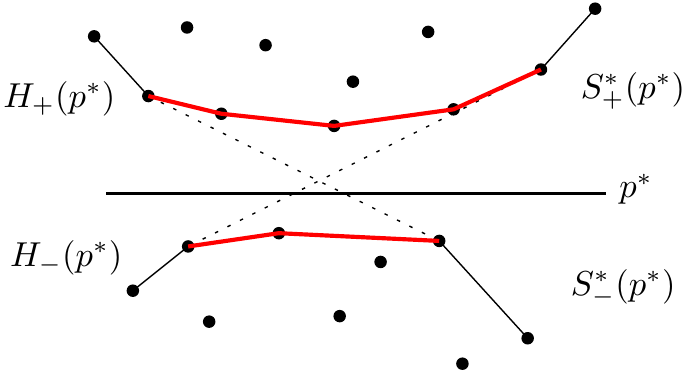}
\caption{\footnotesize Illustrating the dual plane: The (red) thick edges between the two inner common tangents (the dotted segments) are dual to $F_p(S)$.}
\label{fig:dual}
\end{center}
\end{minipage}
\vspace{-0.15in}
\end{figure}

Let $S^*_+(p^*)$ denote the subset of $S^*$ above $p^*$ and $S^*_-(p^*)$ the subset of $S^*$ below $p^*$ (note that $p^*$ is not vertical). We use $H_+(p^*)$ to denote the half hull of the convex hull of
$S^*_+(p^*)$ facing $p^*$ (e.g., if $p^*$ is horizontal, then
$H_+(p^*)$ is the lower hull; for this reason, we call $H_+(p^*)$ the {\em lower hull}; see Fig~\ref{fig:dual}); similarly, we use
$H_-(p^*)$ to denote the half hull of the convex hull of $S^*_-(p^*)$
facing $p^*$ and we call it the {\em upper hull}.
According to the above discussion, $F_p(S)$ is dual to
the portions of $H_+(p^*)$ and $H_-(p^*)$ between their inner common
tangents, and we use $F^*_p(S)$ to denote the dual of $F_p(S)$. Our
algorithm to be presented below will implicitly determine $H_+(p^*)$
and $H_-(p^*)$ (more precisely, each of them is maintained in a
binary search tree of height $O(\log n)$
that can support standard binary search in $O(\log n)$ time), after
which their inner common tangents can be computed in $O(\log n)$
time~\cite{ref:GuibasCo91}
and then $F^*_p(S)$ can be output in
additional $O(|F^*_p(S)|)$ time. Again, if $F^*_p(S)$ is the same for
multiple points $p\in P$, then $F^*_p(S)$ will be output only once.
In the following, depending on the context, a convex hull (resp.,
upper hull, lower hull) may refer to a binary search tree that
represents it. For example, ``computing $H_+(p^*)$'' means
``computing a binary search tree that represents $H_+(p^*)$''.

We compute a hierarchical $(1/r)$-cutting $\Xi_0,\Xi_1,\ldots,\Xi_k$
for the lines of $P^*$ with $k$ and a constant $\rho$ as defined in
	Section~\ref{sec:pre}, and with $r$ to be determined later, along
	with the subsets $P^*_\sigma$ of lines of $P^*$ crossing $\sigma$
	for all cells $\sigma$ of $\Xi_i$ for all $i=0,1,\ldots,k$.
This can be done in $O(mr)$ time~\cite{ref:ChazelleCu93}. Recall that $k=O(\log r)$.
For each point $l^*\in S^*$, we find the cell $\sigma\in \Xi_i$ containing $l^*$ for all $i=0,1,\ldots,k$ and store $l^*$ in the set $S^*(\sigma)$, i.e., $S^*(\sigma)=S^*\cap \sigma$.
Computing the sets $S^*(\sigma)$ for all cells $\sigma\in \Xi_i$, $i=0,1,\ldots,k$, takes $O(n\log r)$ time~\cite{ref:ChazelleCu93}.
As each point of $S^*$ is stored in a single cell of $\Xi_i$, for each $0\leq i\leq k$, the total size of $S^*(\sigma)$ for all cells $\sigma$ of the cutting is $O(n\log r)$.
If initially we sort all points of $S^*$ by $x$-coordinate, then we can obtain the sorted lists of all sets $S^*(\sigma)$ of all cells in $O(n\log r)$ time in total. Using the sorted lists, for each cell $\sigma\in \Xi_i$, $i=0,1,\ldots,k$, we compute the convex hull of $S^*(\sigma)$ in $O(|S^*(\sigma)|)$ time (and store it in a balanced binary search tree). All above takes $O(mr+n\log r+n\log n)$ time in total.

Next, for each cell $\sigma$ of the last cutting $\Xi_k$, if $|S^*(\sigma)|>n/r^2$, then we further triangulate $\sigma$ (which itself is a triangle) into $\Theta(|S^*(\sigma)|\cdot r^2/n)$ triangles each of which contains at most $n/r^2$ points of $S^*$. As points of $S^*(\sigma)$ are already sorted by $x$-coordinate, the triangulation can be easily done in $O(|S^*(\sigma)|)$ time, as follows. By sweeping the points of $S^*(\sigma)$ from left to right, we can partition $\sigma$ in to $\lceil |S^*(\sigma)|\cdot r^2/n \rceil$ trapezoids each of which contains no more than $n/r^2$ points of $S^*$. Then, we partition each trapezoid into two triangles. In this way, $\sigma$ is triangulated into at most $2\lceil |S^*(\sigma)|\cdot r^2/n \rceil$ triangles each containing at most $n/r^2$ points of $S^*$.
Processing all cells of $\Xi_k$ as above takes $O(n)$ time in total.
For convenience, we use $\Xi_{k+1}$ to refer to the set of all new triangles obtained above. Since $\Xi_k$ has $O(r^2)$ cells, by our way of computing the triangles of $\Xi_{k+1}$, the size of $\Xi_{k+1}$ is bounded by $O(r^2)$. For each triangle $\sigma'\in \Xi_{k+1}$, if $\sigma'$ is in the cell $\sigma$ of $\Xi_k$, then we also say that $\sigma$ is the {\em parent} of $\sigma'$ and $\sigma'$ is a {\em child} of $\sigma$ (note that the number of children of $\sigma$ may not be $O(1)$).
We also define $S^*(\sigma')=S^*\cap \sigma'$, and compute and store the convex hull of $S^*(\sigma')$. This takes $O(n)$ time for all triangles $\sigma'$ of $\Xi_{k+1}$, thanks to the presorting of $S^*$.

For reference purpose, we consider the above {\em the preprocessing step} of our algorithm.

For each line $p^*\in P^*$, we process it as follows. Without loss of generality, we assume that $p^*$ is horizontal. Let $\Psi(p^*)$ denote the set of all cells $\sigma$ of $\Xi_i$ crossed by $p^*$, for all $i=0,1,\ldots,k$. Let $\Psi_{k+1}(p^*)$ denote the set of all cells $\sigma$ of $\Xi_{k+1}$ crossed by $p^*$. For each cell $\sigma\in \Psi_{k+1}(p^*)$, we use $\sigma_+(p^*)$ to denote the portion of $\sigma$ above $p^*$, and let $S^*(\sigma_+(p^*))=S^*\cap \sigma_+(p^*)$. Next we define a set $\Psi_+(p^*)$ of cells of $\Xi_{i}$, $i=0,1,\ldots,k+1$.
For each cell $\sigma'\in \Psi(p^*)$, suppose $\sigma'$ is in $\Xi_i$ for some $i\in [0,k]$. For each child $\sigma$ of $\sigma'$ (thus $\sigma\in \Xi_{i+1}$), if $\sigma$ is completely above the line $p^*$, then $\sigma$ is in  $\Psi_+(p^*)$.
We have the following lemma.

\begin{lemma}\label{lem:10}
$S^*_+(p^*)$ is the union of $\bigcup_{\sigma\in \Psi_+(p^*)}S^*(\sigma)$ and $\bigcup_{\sigma\in \Psi_{k+1}(p^*)}S^*(\sigma_+(p^*))$.
\end{lemma}
\begin{proof}
First of all, by definition, all points of $S^*(\sigma)$ for all cells $\sigma\in \Psi_+(p^*)$ are above $p^*$ and thus are in $S^*_+(p^*)$; similarly, all points of $S^*(\sigma_+(p^*))$ for all cells $\sigma\in \Psi_{k+1}(p^*)$ are above $p^*$ and thus are in $S^*_+(p^*)$. Hence, both $\bigcup_{\sigma\in \Psi_+(p^*)}S^*(\sigma)$ and $\bigcup_{\sigma\in \Psi_{k+1}(p^*)}S^*(\sigma_+(p^*))$ are subsets of $S^*_+(p^*)$.

On the other hand, consider a point $l^*\in S^*_+(p^*)$. By definition, $l^*$ is above the line $p^*$. It suffices to prove that $l^*$ must be in either $\bigcup_{\sigma\in \Psi_+(p^*)}S^*(\sigma)$ or $\bigcup_{\sigma\in \Psi_{k+1}(p^*)}S^*(\sigma_+(p^*))$. If $l^*$ is in a cell $\sigma$ of $\Psi_{k+1}$ that is crossed by $p^*$ (thus $\sigma\in \Psi_{k+1}(p^*)$), then since $l^*$ is above $p^*$, $l^*$ must be in $S^*(\sigma_+(p^*))$ and thus is in $\bigcup_{\sigma\in \Psi_{k+1}(p^*)}S^*(\sigma_+(p^*))$. Otherwise, $l^*$ is not in any cell of $\Xi_{k+1}$ crossed by $p^*$. Hence, $l^*$ must be in a cell $\sigma\in \Xi_{i+1}$ that is not crossed by $p^*$ but whose parent cell $\sigma'\in \Xi_i$ is crossed by $p^*$, for some $i\in [0,k]$, because $\Xi_0$, which consists of a single cell that is the entire plane, is crossed by $p^*$. As such, $\sigma'$ is in $\Psi(p^*)$. Further, since $\sigma$ is not crossed by $p^*$ and $\sigma$ contains $l^*$, which is above $p^*$, $\sigma$ must be completely above $p^*$. Therefore, $\sigma$ must be in $\Psi_+(p^*)$, and thus $l^*$ is in $\bigcup_{\sigma\in \Psi_+(p^*)}S^*(\sigma)$.
\end{proof}

Lemma~\ref{lem:10} implies that if we have convex hulls of $S^*(\sigma)$ for all cells $\sigma\in \Psi_+(p^*)$ and convex hulls of $S^*(\sigma_+(p^*))$ for all cells $\sigma\in \Psi_{k+1}(p^*)$, then $H_+(p^*)$ is the lower hull of all these convex hulls. Define $\calH_+(p^*)$ as the set of convex hulls mentioned above. Thanks to our preprocessing step, we have the following lemma.

\begin{lemma}\label{lem:20}
We can obtain (binary search trees representing) the convex hulls of $\calH_+(p^*)$ for all lines $p^*\in P^*$ in $O(mr+mn/r)$ time.
\end{lemma}
\begin{proof}
First of all, the sets $\Psi(p^*)$ for all lines $p^*\in P^*$ can be obtained in $O(mr)$ time when we compute the cutting~\cite{ref:ChazelleCu93}. This also means that the total size $|\Psi(p^*)|$ for all $p^*\in P^*$ is $O(mr)$; as a matter of fact, the total size $|\Psi(p^*)|$ for all $p^*\in P^*$ is equal to $\sum_{i=0}^k\sum_{\sigma\in \Xi_i}|P^*_{\sigma}|$, which is $O(mr)$~\cite{ref:ChazelleCu93}.

For each cell $\sigma'\in \Psi(p^*)$, suppose $\sigma'$ is in $\Xi_i$, for some $i\in [0,k]$.
We check every child $\sigma$ (in $\Xi_{i+1}$) of $\sigma'$ to determine whether it is in $\Psi_+(p^*)$.
This computes the set $\Psi_+(p^*)$. For
the runtime,
since $\sigma'$ has $O(1)$ children $\sigma\in \Xi_{i+1}$ for $i<k$,
all cells $\sigma$ of $\Psi_+(p^*)$ that are not in $\Xi_{k+1}$ can
be obtained in a total of $O(mr)$ time for all $p^*\in P^*$. For the
time we spend on computing cells of $\Psi_+(p^*)$ that are in
$\Xi_{k+1}$, observe that each cell $\sigma$ of $\Xi_{k+1}$ will be
checked $t$ times in the entire algorithm for all $p^*\in P^*$, where
$t$ is the number of lines of $P^*$ crossing the parent $\sigma'\in
\Xi_k$ of $\sigma$. According to the property of the cutting $\Xi_k$, $t=O(m/r)$. Hence, each cell $\sigma$ of $\Xi_{k+1}$ will be checked $O(m/r)$ times in the entire algorithm. As $\Xi_{k+1}$ has $O(r^2)$ cells, the total time for finding cells of $\Psi_+(p^*)$ that are in $\Xi_{k+1}$ is bounded by $O(mr)$.
As such, computing the set $\Psi_+(p^*)$ takes $O(mr)$ time. Recall that for each cell $\sigma\in \Xi_i$, $i=0,1,\ldots,k$, a binary search tree representing the convex hull of $S^*(\sigma)$ has been computed in the preprocessing step. Hence, convex hulls of all cells of $\Psi_+(p^*)$ are available.

We proceed to compute the convex hulls of $S^*(\sigma_+(p^*))$ for all cells $\sigma\in \Psi_{k+1}(p^*)$. We first compute the set $\Psi_{k+1}(p^*)$, using an algorithm similar to above. For each cell $\sigma'\in \Psi(p^*)$, which is already computed above, if $\sigma'$ is in $\Xi_k$, then we check every child of $\sigma'$ and determine whether it is in $\Psi_{k+1}(p^*)$. In this way, the sets $\Psi_{k+1}(p^*)$ for all lines $p^*\in P^*$ can be computed in $O(mr)$ time. This also implies that $\sum_{p^*\in P^*}|\Psi_{k+1}(p^*)|=O(mr)$. Next, for each $p^*\in P^*$, for each cell $\sigma\in \Psi_{k+1}(p^*)$, we compute the convex hull of $S^*(\sigma_+(p^*))$, which can be done in $O(|S^*(\sigma)|)$ time since points of $S^*(\sigma)$ are already sorted. As $|S^*(\sigma)|\leq n/r^2$ for all cells $\sigma\in \Xi_{k+1}$, $\Psi_{k+1}(p^*)\subseteq \Xi_{k+1}$, and $\sum_{p^*\in P^*}|\Psi_{k+1}(p^*)|=O(mr)$, the total time for computing the convex hulls of $S^*(\sigma_+(p^*))$ for all cells $\sigma\in \Psi_{k+1}(p^*)$ for all lines $p^*\in P^*$ is $O(mr\cdot n/r^2)$, which is $O(mn/r)$.

In summary, binary search trees representing convex hulls of $\calH_+(p^*)$ for all lines $p^*\in P^*$ can be obtained in $O(mr+mn/r)$ time in total.
\end{proof}

With the preceding lemma, our next goal is to compute the lower hull of all convex hulls of $\calH_+(p^*)$. To this end, the observation in the following lemma is critical.

\begin{lemma}\label{lem:30}
For each $p^*\in P^*$, the convex hulls of $\calH_+(p^*)$ are pairwise disjoint.
\end{lemma}
\begin{proof}
According to the definition of $\calH_+(p^*)$, each convex hull is inside a cell either in $\Psi_+(p^*)$ or in $\Psi_{k+1}(p^*)$. Hence, to prove the lemma, it suffices to show that cells of $\Psi_+(p^*)\bigcup \Psi_{k+1}(p^*)$ are pairwise disjoint. Note that each cell of $\Psi_+(p^*)\bigcup \Psi_{k+1}(p^*)$ is a cell in $\Xi_i$, for some $i\in [0,k+1]$.

Assume to the contrary that two different cells $\sigma_1$ and $\sigma_2$ of $\Psi_+(p^*)\bigcup \Psi_{k+1}(p^*)$ are not disjoint. Then, by the properties of the hierarchical cutting, one of $\sigma_1$ and $\sigma_2$ must be an ancestor of the other. Without loss of generality, we assume that $\sigma_1$ is an ancestor of $\sigma_2$, and thus $\sigma_2$ is contained in $\sigma_1$. Suppose $\sigma_1$ is in $\Xi_i$ and $\sigma_2$ is in $\Xi_j$ with $0\leq i<j\leq k+1$. As $i<k+1$, by definition, $\sigma_1$ cannot be in $\Psi_{k+1}(p^*)$ and thus must be in $\Psi_+(p^*)$. Hence, $\sigma_1$ must be completely above the line $p^*$. As $\sigma_2$ is contained in $\sigma_1$, $\sigma_2$ is also above $p^*$. As such, $\sigma_2$ cannot be in $\Psi_{k+1}(p^*)$ and thus is in $\Psi_+(p^*)$. By the definition of $\Psi_+(p^*)$, the line $p^*$ must cross the parent cell $\sigma'$ of $\sigma_2$. On the other hand, since $\sigma_1$ is an ancestor of $\sigma_2$ and $\sigma_1\neq \sigma_2$, either $\sigma_1=\sigma'$ or $\sigma_1$ is an ancestor of $\sigma'$. In either case, $\sigma_1$ must contain $\sigma'$. Since $\sigma_1$ is above the line $p^*$, we obtain that $\sigma'$ is also above $p^*$. This incurs contradiction because $p^*$ crosses $\sigma'$.
\end{proof}

With the convex hulls computed in Lemma~\ref{lem:20} and the property in Lemma~\ref{lem:30}, the next lemma computes the lower hull $H_+(p^*)$ for all $p^*\in P^*$.

\begin{lemma}\label{lem:40}
For each $p^*\in P^*$, suppose the convex hulls of $\calH_+(p^*)$ are available; then we can compute (a binary search tree representing) the lower hull $H_+(p^*)$ in $O(|\calH_+(p^*)|\log n)$ time.
\end{lemma}
\begin{proof}
Without loss of generality, we assume that $p^*$ is horizontal. Let $t=|\calH_+(p^*)|$. Note that the size of each convex hull of $\calH_+(p^*)$ is at most $n$. Also, since convex hulls of $\calH_+(p^*)$ are pairwise disjoint by Lemma~\ref{lem:30}, it holds that $t\leq n$.

Because we are to compute the lower hull of the convex hulls of $\calH_+(p^*)$, it suffices to only consider the lower hull of each convex hull of $\calH_+(p^*)$. Note that since binary search trees for all convex hulls of $\calH_+(p^*)$ are available, we can obtain binary search trees representing their lower hulls in $O(t\log n)$ time by first finding the leftmost and rightmost vertices of the convex hulls and then performing split/merge operations on the trees.

The first step is to compute the portions of each lower hull $H$ of $\calH_+(p^*)$ that is vertically visible to $p^*$ (we say that a point $q\in H$ is {\em vertically visible} to $p^*$ if the vertical segment connecting $q$ to $p^*$ does not cross any other lower hull of $\calH_+(p^*)$). In fact, the visible portions constitute exactly the lower envelope of the lower hulls of $\calH_+(p^*)$, denoted by $\calL(\calH_+(p^*))$ (e.g., see Fig.~\ref{fig:lowenvelope}).
Below we describe an algorithm to compute $\calL(\calH_+(p^*))$ in $O(t\log n)$ time.

\begin{figure}[t]
\begin{minipage}[t]{\textwidth}
\begin{center}
\includegraphics[height=1.3in]{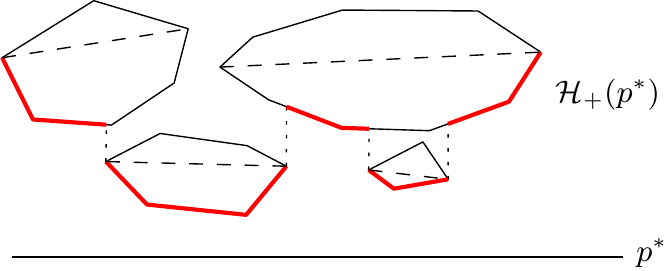}
\caption{\footnotesize Illustrating the lower envelope $\calL(\calH_+(p^*))$ (the thick red edges) of the convex hulls of $\calH_+(p^*)$. The dashed segment inside each convex hull is the representative segment.}
\label{fig:lowenvelope}
\end{center}
\end{minipage}
\vspace{-0.15in}
\end{figure}

 As convex hulls of $\calH_+(p^*)$ are pairwise disjoint by Lemma~\ref{lem:30}, the number of (maximal) visible portions of all lower hulls of $\calH_+(p^*)$ is at most $2t-1$.
For each convex hull $H$ of $\calH_+(p^*)$, consider the segment connecting the leftmost and rightmost endpoints of $H$, and call it the {\em representative segment} of $H$ (e.g., see Fig.~\ref{fig:lowenvelope}). Let $Q$ be the set of representative segments of all convex hulls of $\calH_+(p^*)$. Because convex hulls of $\calH_+(p^*)$ are pairwise disjoint,
an easy but crucial observation is that segments of $Q$ are pairwise disjoint and the lower envelope $\calL(Q)$ of the segments of $Q$ corresponds to $\calL(\calH_+(p^*))$ in the following sense: if $\overline{ab}$ is a maximal segment of $\calL(Q)$ that lies on a representative segment of a convex hull $H$ of $\calH_+(p^*)$, then the vertical projection of $\overline{ab}$ onto the lower hull of $H$ is a maximal portion of the lower hull of $H$ on $\calL(\calH_+(p^*))$, and that portion can be obtained in $O(\log n)$ time by splitting the binary search tree for the lower hull of $H$ at the $x$-coordinates of $a$ and $b$, respectively. As such, once $\calL(Q)$ is computed, $\calL(\calH_+(p^*))$ in which each maximal portion is represented by a binary search tree can be obtained in additional $O(t\log n)$ time. As $|Q|=t$ and segments of $Q$ are pairwise disjoint, $\calL(Q)$ can be constructed in $O(t\log t)$ time by an easy plane sweeping algorithm.
Hence, $\calL(\calH_+(p^*))$ can be computed in $O(t\log n)$ time in total.


With $\calL(\calH_+(p^*))$ in hand, we can now compute the lower hull $H_+(p^*)$ in additional $O(t\log n)$ time, as follows. As discussed above, $\calL(\calH_+(p^*))$ consists of at most $2t-1$ {\em pieces} sorted from left to right, each of which is a portion of a lower hull of $\calH_+(p^*)$ and is represented by a binary search tree. We merge the first two pieces by computing their common tangent, which can be done in $O(\log n)$ time~\cite{ref:OvermarsMa81} as the two pieces are separated by a vertical line. After the merge, we obtain a binary search tree that represents the lower hull of the first two pieces of $\calL(\calH_+(p^*))$. Next, we merge this lower hull with the third piece of $\calL(\calH_+(p^*))$ in the same way. We repeat this process until all pieces of $\calL(\calH_+(p^*))$ are merged, after which a binary search tree representing $H_+(p^*)$ is obtained. The runtime is bounded by $O(t\log n)$ as each merge takes $O(\log n)$ time and $\calL(\calH_+(p^*))$ has at most $2t-1$ pieces.

In summary, once the convex hulls of $\calH_+(p^*)$ are available, we can compute the lower hull $H_+(p^*)$ in $O(|\calH_+(p^*)|\log n)$ time.
\end{proof}

Applying Lemma~\ref{lem:40} to all lines of $P^*$ will compute the lower hulls $H_+(p^*)$ for all $p^*\in P^*$. One issue is that after we apply the algorithm for one line $p^*\in P^*$, convex hulls of $\calH_+(p^*)$ may have been destroyed due to the split and merge operations during the algorithm. The destroyed convex hulls may be used later when we apply the algorithm for other lines of $P^*$. The remedy is to use fully persistent binary search trees with path-copying~\cite{ref:DriscollMa89,ref:SarnakPl86} to represent convex hulls so that standard operations on the trees (e.g., merge, split) can be performed in $O(\log n)$ time each and after each operation the original trees are still kept intact (so that future operations can still be performed on the original trees as usual). In this way, whenever we apply the algorithm for a line of $P^*$, we always have the original trees representing the convex hulls available, and thus the runtime of the algorithm in Lemma~\ref{lem:40} is not affected (although $O(\log n)$ extra space will be incurred after each operation on the trees).

For the time analysis, by Lemma~\ref{lem:20}, computing convex hulls of $\calH_+(p^*)$ for all lines $p^*\in P^*$ takes $O(mr+mn/r)$ time. Then, applying Lemma~\ref{lem:40} to all lines of $P^*$ takes $O(\sum_{p^*\in P^*}|\calH_+(p^*)|\cdot \log n)$ time in total, which is bounded by $O(mr\log n)$ due to the following lemma.

\begin{lemma}
$\sum_{p^*\in P^*}|\calH_+(p^*)|=O(mr)$.
\end{lemma}
\begin{proof}
The lemma actually has been implied by the time analysis of Lemma~\ref{lem:20}. We provide a direct proof here. Notice that $|\calH_+(p^*)|=|\Psi_+(p^*)|+|\Psi_{k+1}(p^*)|$. Below we will bound both $\sum_{p^*\in P^*}|\Psi_+(p^*)|$ and $\sum_{p^*\in P^*}|\Psi_{k+1}(p^*)|$.

Recall that $\Psi_{k+1}(p^*)\subseteq \Xi_{k+1}$. Consider a cell $\sigma\in \Xi_{k+1}$. If $\sigma\in \Psi_{k+1}(p^*)$, then $\sigma'$ must be crossed by the line $p^*$, where $\sigma'$ is the parent cell of $\sigma$ in $\Xi_k$. Hence, the number of lines $p^*\in P^*$ such that $\sigma$ is in $\Psi_{k+1}(p^*)$ is no more than the number lines of $P^*$ crossing $\sigma'$, which is $O(m/r)$. As $\Xi_{k+1}$ has $O(r^2)$ cells, $\sum_{p^*\in P^*}|\Psi_{k+1}(p^*)|$ is bounded by $O(r^2\cdot m/r)$, which is $O(mr)$.

To bound $\sum_{p^*\in P^*}|\Psi_+(p^*)|$, we partition $\Psi_+(p^*)$ into two subsets, $\Psi^1_+(p^*)$, which consists of those cells of $\Psi_+(p^*)$ that are in $\Xi_{k+1}$, and $\Psi^2_+(p^*)=\Psi_+(p^*)\setminus \Psi^1_+(p^*)$.

\begin{itemize}
  \item For each cell $\sigma\in \Psi^1_+(p^*)$, by definition, $\sigma'$ must be crossed by the line $p^*$, where $\sigma'$ is the parent cell of $\sigma$ in $\Xi_k$. Hence, the number of lines $p^*\in P^*$ such that $\sigma$ is in $\Psi^1_+(p^*)$ is no more than the number lines of $P^*$ crossing $\sigma'$, which is $O(m/r)$. Following the same analysis as above, we can derive that $\sum_{p^*\in P^*}|\Psi^1_+(p^*)|=O(mr)$.
  \item
  For each cell $\sigma\in \Psi^2_+(p^*)$, $\sigma$ is in $\Xi_i$ for some $i\in [1,k]$. By definition, $\sigma'$ must be crossed by the line $p^*$, where $\sigma'$ is the parent cell of $\sigma$ in $\Xi_{i-1}$. Hence, the number of lines $p^*\in P^*$ such that $\sigma$ is in $\Psi^1_+(p^*)$ is no more than the number lines of $P^*$ crossing $\sigma'$, which is $O(m/\rho^{i-1})$, where $\rho$ is the constant associated with the hierarchical cutting as explained before. As $\Xi_i$ has $O(\rho^{2i})$ cells, $\sum_{p^*\in P^*}|\Psi^1_+(p^*)|$ is big-O of $\sum_{1\leq i\leq k}m/\rho^{i-1}\cdot \rho^{2i}=\sum_{1\leq i\leq k}m\rho^{i+1}$, which is $O(mr)$ as $k=\Theta(\log_{\rho}r)$ and $\rho$ is a constant. As such, $\sum_{p^*\in P^*}|\Psi^1_+(p^*)|=O(mr)$.
\end{itemize}

Therefore, we obtain that $\sum_{p^*\in P^*}|\Psi_+(p^*)|=O(mr)$.
\end{proof}

In summary, computing lower hulls $H_+(p^*)$ for all $p^*\in P^*$ can be done in a total of $O(n\log n+n\log r+mr\log n+mn/r)$ time. Analogously, we can also compute the upper hulls $H_-(p^*)$ for all $p^*\in P^*$. Then, for each line $p^*\in P^*$, we compute the two inner common tangents of $H_+(p^*)$ and $H_-(p^*)$, which can be done in $O(\log n)$ time~\cite{ref:GuibasCo91}. With the two inner common tangents as well as the two hulls $H_+(p^*)$ and $H_-(p^*)$, the dual face $F_p^*(S)$, or the face $F_p(S)$ in the primal plane, can be implicitly determined. More precisely, given $H_+(p^*)$ and $H_-(p^*)$, we can obtain a binary search tree representing $F_p(S)$ in $O(\log n)$ time. The tree can be used to support standard binary search on $F_p(S)$, which is a convex polygon. Outputting $F_p(S)$ explicitly takes $O(|F_p(S)|)$ additional time.

To avoid reporting a face more than once, we can remove duplication in the following way. Due to the general position assumption, an easy observation is that $F_{p_1}(S)=F_{p_2}(S)$ for two points $p_1$ and $p_2$ of $P$ if and only if the leftmost vertex of $F_{p_1}(S)$ is the same as that of $F_{p_2}(S)$. Also note that the leftmost and rightmost vertices of $F_p(S)$ are dual to the two inner common tangents of $H_+(p^*)$ and $H_-(p^*)$, respectively. Hence, for any two points $p_1$ and $p_2$ of $P$, we can determine whether they are from the same face of $\calA(S)$ by comparing the corresponding inner common tangents. In this way, the duplication can be removed in $O(m\log m)$ time, which is $O(m\log n)$ as $m<n^2/2$. After that, we can report all distinct faces. Note that outputting all distinct faces explicitly takes $O(n+m\sqrt{n})$ time as the total combinatorial complexity of $m$ distinct faces in $\calA(S)$ is bounded by $O(n+m\sqrt{n})$~\cite{ref:EdelsbrunnerOn86}.

To recapitulate, computing the distinct faces $F_p(S)$ implicitly for all $p\in P$ takes $O(n\log n+n\log r+mr\log n+mn/r)$ time and reporting them explicitly takes additional $O(n+m\sqrt{n})$ time. Setting $r=\min\{m,\sqrt{n/\log n}\}$ leads to the total time bounded by $O(n\log n+m\sqrt{n\log n})$.

\begin{theorem}\label{theo:10}
Given a set $S$ of $n$ lines and a set $P$ of $m$ points in the plane,
the faces of the arrangement of the lines containing at least one point of $P$ can be computed in $O(n\log n+m\sqrt{n\log n})$ time.
\end{theorem}

\paragraph{Remark.}
The algorithm runs in $O(n\log n)$ for $m=O(\sqrt{n\log n})$, which matches the $\Omega(n\log n)$ lower bound for computing a single face. For comparison, Agarwal~\cite{ref:AgarwalPa902} gave an algorithm of $O(n\log n+m\sqrt{n}\log^2 n)$ time, and Agarwal, Matou\v{s}ek, and Schwarzkopf~\cite{ref:AgarwalCo98} presented a randomized algorithm of $O(n\log n+m\sqrt{n}\log n)$ expected time.


\subsection{The second algorithm}
\label{sec:second}

We now present our main algorithm, which
follows the scheme of Agarwal~\cite{ref:AgarwalPa902}, but replaces a key subroutine by Theorem~\ref{theo:10}.

We first compute a $(1/r)$-cutting $\Xi$ for the lines of $S$ in
$O(nr)$ time~\cite{ref:ChazelleCu93}, with the parameter $r$ to be
determined later. We then locate the cell of $\Xi$ containing each point of
$P$; this can be done in $O(m\log r)$ time for all points of
$P$~\cite{ref:ChazelleCu93}. Consider a cell $\sigma$ of $\Xi$. Recall
that $\sigma$ is a triangle. Let $P(\sigma)=P\cap \sigma$. Let
$S_{\sigma}$ be the subset of lines of $S$ crossing $\sigma$.
Consider a point $p\in P(\sigma)$.
Recall the definition in Section~\ref{sec:pre} that $F_p(S_{\sigma})$ denotes the face of the arrangement  $\calA(S_{\sigma})$ that contains $p$.
Observe that the face $F_p(S)$ is $F_{p}(S_{\sigma})$ if and only $F_{p}(S_{\sigma})$ does
not intersect the boundary of $\sigma$. The {\em zone} of $\sigma$ in
$\calA(S_{\sigma})$ is defined as the collection of face portions
$F\cap \sigma$ for all faces $F\in \calA(S_{\sigma})$ that intersect
the boundary of $\sigma$. If $F_p(S)\neq F_{p}(S_{\sigma})$, then $F_p(S)$ is
divided into multiple portions each of which is a face in the zone of
some cell of $\Xi$ (and $F_{p}(S_{\sigma})$ is one of these portions).
Hence, to find all nonempty faces of $\calA(S)$, it
suffices to compute, for every cell $\sigma\in \Xi$, the faces of
$\calA(S_{\sigma})$ containing the points of $P(\sigma)$ and the zone
of $\sigma$. The nonempty faces of $\calA(S)$ that are split among the
zones can be obtained by merging the zones along the edges of cells
of $\Xi$.

To compute the faces of $\calA(S_{\sigma})$ containing the points of $P(\sigma)$, we apply Theorem~\ref{theo:10}, which takes $O(n_{\sigma}\log n_{\sigma}+m_{\sigma}\sqrt{n_{\sigma}\log n_{\sigma}})$ time, with $n_{\sigma}=|S_{\sigma}|$ and $m_{\sigma}=|P_{\sigma}|$.
Computing the zone for $\sigma$ can be done in $O(n_{\sigma}\log n_{\sigma})$ time, e.g., by the algorithm of~\cite{ref:AlevizosAn90} or a recent simple algorithm~\cite{ref:WangA22}. Since $n_{\sigma}=O(n/r)$, $\sum_{\sigma\in \Xi}m_{\sigma}=m$, and $\Xi$ has $O(r^2)$ cells, the total time for solving the subproblem for all cells of $\Xi$ is
\begin{equation*}
\begin{split}
   O\left(\sum_{\sigma\in \Xi}n_{\sigma}\log n_{\sigma}+m_{\sigma}\sqrt{n_{\sigma}\log n_{\sigma}}\right) &= O\left( r^2\cdot (n/r)\cdot \log (n/r)+\sqrt{n/r\cdot \log(n/r)}\cdot \sum_{\sigma\in \Xi}m_{\sigma}\right)\\
     & =O\left(nr\log (n/r)+m\sqrt{n/r\cdot \log(n/r)}\right).
\end{split}
\end{equation*}

After all cells of $\Xi$ are processed as above, we merge the zones of all cells, which can be done in time linear in the total size of the zones of all cells of $\Xi$ because zones of different cells are disjoint. The total size of all zones is $O(nr)$ as the size of the zone for each cell is $O(n/r)$ and $\Xi$ has $O(r^2)$ cells.

In summary, the total time of the algorithm is $O(m\log r+nr\log
(n/r)+m\sqrt{n/r\cdot \log(n/r)})$. By setting
$r=\max\{m^{2/3}/(n^{1/3}\cdot \log^{1/3}(n/\sqrt{m})),1\}$, we obtain that the total time is bounded by $O((nm\log(n/\sqrt{m}))^{2/3}+(n+m)\log n)$. Indeed, since
$m<n^2/2$, $m^{2/3}/(n^{1/3}\cdot \log^{1/3}(n/\sqrt{m}))<n$ and $\log m=O(\log n)$. If $m^{2/3}/(n^{1/3}\cdot \log^{1/3}(n/\sqrt{m}))<1$, then $r=1$ and $m<\sqrt{n\log n}$, and thus $m\log r+nr\log
(n/r)+m\sqrt{n/r\cdot \log(n/r)}=m+n\log n+m\sqrt{n\log n}=O(n\log n)$; otherwise, $r=m^{2/3}/(n^{1/3}\cdot \log^{1/3}(n/\sqrt{m}))$ and $m\log r+nr\log
(n/r)+m\sqrt{n/r\cdot \log(n/r)}=O((nm\log(n/\sqrt{m}))^{2/3}+m\log n)$. Combining with the $O(m\log n)$ time algorithm for the case $m\geq n^2/2$, we obtain the following result.

\begin{theorem}\label{theo:20}
Given a set $S$ of $n$ lines and a set $P$ of $m$ points in the plane,
the faces of the arrangement of the lines containing at least one point of $P$ can be computed in 
$O(n^{2/3}m^{2/3}\log^{2/3}\frac{n}{\sqrt{m}}+(n+m)\log n)$ time.
\end{theorem}

\section{Computing many cells in arrangements of segments}
\label{sec:segment}

In this section, we consider the segment case for computing many faces. Let $S$ be a set of $n$ line segments and $P$ be a set of $n$ points in the plane. The problem is to compute all distinct non-empty faces of $\calA(S)$. Note that these faces will be output explicitly.
For ease of exposition, we make a general position assumption that no segment of $S$ is vertical, no three segments of $S$ are concurrent, no two segments of $S$ share a common endpoint, and no point of $P$ lies on a segment of $S$. Degenerate cases can be handled by standard techniques~\cite{ref:EdelsbrunnerSi90}.

In the following, we first present our deterministic algorithm and then give the randomized result, which uses the deterministic algorithm as a subroutine.

\subsection{The deterministic algorithm}

If $m\geq n^2/2$, then the problem can be solved in $O(m\log n)$ time using the straightforward algorithm mentioned in Section~\ref{sec:pre} (i.e., first compute $\calA(S)$ and then find the non-empty cells using point locations). In what follows, we assume that $m<n^2/2$, and thus $\log m=O(\log n)$.
Let $E$ denote the set of the endpoints of all segments of $S$. Let $L$ denote the set of supporting lines of all segments of $S$.

Initially, we sort the points of $E$ (resp., $P$) by $x$-coordinate.
We compute a $(1/r)$-cutting $\Xi$ for $L$ in $O(nr)$ time~\cite{ref:ChazelleCu93}, for a sufficiently large constant $r$. We then locate the cell of $\Xi$ containing each point of
$P$; this can be done in $O(m\log r)$ time for all points of
$P$~\cite{ref:ChazelleCu93}. Consider a cell $\sigma$ of $\Xi$. Recall
that $\sigma$ is a triangle. Let $P(\sigma)=P\cap \sigma$ and $E(\sigma)=E\cap \sigma$. Let
$S_{\sigma}$ be the subset of segments of $S$ intersecting $\sigma$.
Note that $|S_{\sigma}|=O(n/r)$ and $\Xi$ has $O(r^2)$ cells.
If $|E(\sigma)|>n/r^2$, then we triangulate $\sigma$ into
at most $2\lceil|E(\sigma)|\cdot r^2/n\rceil$ triangles each of which contains at
most $n/r^2$ points of $E$. This can be done by first sorting all points of
$E(\sigma)$ and then using a sweeping algorithm as described in
Section~\ref{sec:first}. Due to the presorting of $E$, the sorting of
$E(\sigma)$ for all cells $\sigma$ of $\Xi$ can be done in $O(n)$ time
and thus the triangulation takes $O(n)$ time in total for all cells of
$\Xi$. By slightly abusing notation, we still use $\Xi$ to denote the
set of all new triangles for all original cells (if an original cell was not triangulated, then we also include it in the new $\Xi$). The new $\Xi$ now has the following properties: each cell of $\Xi$ is intersected by
$O(n/r)$ segments of $S$, $\Xi$ has $O(r^2)$ cells, and each cell of
$\Xi$ contains at most $n/r^2$ points of $E$.


For each cell
$\sigma\in \Xi$, if $|P(\sigma)|>m/r^2$, then we further triangulate $\sigma$
into at most $2\lceil|P(\sigma)|\cdot r^2/m\rceil$ triangles each of which contains at
most $m/r^2$ points of $P$. Due to the presorting of $P$, the triangulation can be done in $O(m)$ time in total for all cells of $\Xi$, in the same way as above for $E(\sigma)$.
By slightly abusing notation, we still use $\Xi$ to denote the
set of all new triangles. The new $\Xi$ now has the following properties: each cell of $\Xi$ is intersected by $O(n/r)$ segments of $S$, $\Xi$ has $O(r^2)$ cells, each cell of
$\Xi$ contains at most $n/r^2$ points of $E$, and each cell of $\Xi$ contains at
most $m/r^2$ points of $P$.

For each cell $\sigma\in \Xi$, we define $S_{\sigma}$, $E(\sigma)$, and $P(\sigma)$ in the same way as before.
We say that a segment of $S_{\sigma}$ is a {\em short segment} of $\sigma$ if it has an
endpoint in the interior of $\sigma$ and is a {\em long segment}
otherwise. Let $S_1(\sigma)$ denote the set of long segments of
$\sigma$ and $S_2(\sigma)$ the set of short segments of
$\sigma$. Since $S_1(\sigma)\subseteq S_{\sigma}$ and $|S_{\sigma}|=O(n/r)$, we have $|S_1(\sigma)|=O(n/r)$. Also note that $|S_2(\sigma)|\leq |E(\sigma)|$.
As $|E(\sigma)|\leq n/r^2$, it holds that $|S_2(\sigma)|\leq n/r^2$.

For each cell edge $e$ of $\Xi$, we define its {\em zone} as the set
of faces of $\calA(S)$ intersected by $e$, which can be computed in
$O(n\alpha^2(n)\log n)$ time~\cite{ref:AmatoCo95} (note that computing
the zone in an arrangement of segments can be reduced to computing a
single face and the size of the zone is $O(n\alpha(n))$~\cite{ref:EdelsbrunnerAr92}).
Consider a point $p\in P(\sigma)$ for any cell
$\sigma\in \Xi$. Recall the definition in Section~\ref{sec:pre} that $F_p(S_{\sigma})$ denotes the face of the arrangement $\calA(S_{\sigma})$ that contains $p$.
If $F_p(S_{\sigma})$ does not intersect any edge of
$\sigma$, then $F_p(S)$ is $F_p(S_{\sigma})$; otherwise, $F_p(S)$ is
a face of the zone of an edge of $\sigma$ (and that face contains
$p$).

In light of the above discussion, our algorithm works as follows.
We first compute the zones for all cell edges of $\Xi$ and explicitly
store them in a point location data structure~\cite{ref:EdelsbrunnerOp86,ref:KirkpatrickOp83}.
This takes $O(nr^2\alpha^2(n)\log n)$ time in total.
Next, for each cell $\sigma\in \Xi$, for each point $p\in P(\sigma)$,
using the point location data structure, we determine in $O(\log n)$ time whether $p$ is in a face of the zone of any edge of
$\sigma$. If yes, we explicitly output the face, which is $F_p(S)$.
Otherwise, the face $F_p(S_{\sigma})$ is $F_p(S)$.
Let $P'(\sigma)$
denote the subset of points $p$ of $P(\sigma)$ in the above second case (i.e., $F_p(S_{\sigma})$ is $F_p(S)$).
The remaining problem is to compute the faces of $\calA(S_{\sigma})$
containing at least one point of $P'(\sigma)$. To solve this subproblem,
observe that the face $F_p(S_{\sigma})$ is in the intersection of $F_p(S_1(\sigma))$ and $F_p(S_2(\sigma))$, which may contain multiple connected components. Hence, more precisely, $F_p(S_{\sigma})$ is the connected component of $F_p(S_1(\sigma))\cap F_p(S_2(\sigma))$ that contains $p$.
Let $L_1(\sigma)$ be the set of the supporting lines of all segments of
$S_1(\sigma)$. Because all segments of $S_1(\sigma)$ are long
segments, we have the following lemma.

\begin{lemma}
For any point $p\in P'(\sigma)$,
$F_p(S_{\sigma})$ is the connected component of
$F_p(L_1(\sigma))\cap F_p(S_2(\sigma))$ that contains $p$.
\end{lemma}
\begin{proof}
Recall that $F_p(S_{\sigma})$ is the connected component of
$F_p(S_1(\sigma))\cap F_p(S_2(\sigma))$ that contains $p$. As $p\in
P'(\sigma)$, we know that $F_p(S_{\sigma})$ is in the interior of $\sigma$. For any
segment $s\in S_1(\sigma)$, suppose that we extend $s$ to a full line.
As $s$ is a long segment, the extension of $s$ does not intersect the
interior of $\sigma$ and thus does not intersect $F_p(S_{\sigma})$. This implies
the lemma.
\end{proof}

Due to the above lemma, to compute the faces of $\calA(S_{\sigma})$
containing the points of $P'(\sigma)$, we do the following:
(1) compute the faces of $\calA(L_1(\sigma))$ containing the points of $P'(\sigma)$;
(2) compute the faces of $\calA(S_2(\sigma))$ containing the points of $P'(\sigma)$;
(3) compute the faces $F_p(S_{\sigma})$ for all points $p\in
P'(\sigma)$ by intersecting the faces obtained in the first two steps
and computing the connected components containing the points of
$P'(\sigma)$. We discuss how to implement the three steps below.

\begin{enumerate}
  \item The first step can be done by applying our algorithm for
the line case in Theorem~\ref{theo:20}, because $L_1(\sigma)$ is a set of lines. As
$|L_1(\sigma)|=|S_1(\sigma)|=O(n/r)$ and $|P'(\sigma)|\leq
|P(\sigma)|\leq m/r^2$, the runtime of the algorithm is bounded by
$$O\left(\frac{n^{2/3}m^{2/3}}{r^2}\log^{2/3}
\frac{n}{\sqrt{m}}+\left(\frac{n}{r}+\frac{m}{r^2}\right)\log
\frac{n}{r}\right).$$

In addition, the total size of all computed faces is
\begin{equation}\label{equ:10}
  O\left(\frac{n^{2/3}m^{2/3}}{r^2}+\frac{n}{r}\right),
\end{equation}
 by applying the upper bound
$O(m^{2/3}n^{2/3}+n)$ on
the combinatorial complexity of $m$ distinct faces in an arrangement
of $n$ lines~\cite{ref:ClarksonCo90}. This bound will be needed later in the time analysis of the third step.

  \item
  For the second step, we apply our
algorithm recursively on $S_2(\sigma)$ and $P'(\sigma)$, so the
problem size becomes $(n/r^2,m/r^2)$ as $|S_2(\sigma)|\leq n/r^2$ and
$|P'(\sigma)|\leq m/r^2$.

Also, the total size of all computed faces
is bounded by
\begin{equation}\label{equ:20}
O\left(\frac{n^{2/3}m^{2/3}}{r^{8/3}}+\frac{n}{r^2}\alpha(\frac{n}{r^2})+\frac{n}{r^2}\log\frac{m}{r^2}\right),
\end{equation}
by applying the $(m^{2/3}n^{2/3}+n\alpha(n)+n\log m)$ upper bound on
the combinatorial complexity of $m$ distinct faces in an arrangement
of $n$ segments~\cite{ref:AronovTh92}. This bound will be needed later in the time analysis of the third step.

\item
The third step can be done by applying the blue-red merge algorithm
of~\cite{ref:EdelsbrunnerTh90}, which takes
$O((\beta+\rho+|P'(\sigma)|)\log(\beta+\rho+|P'(\sigma)|))$
time, where $\beta$ is the total size of all faces computed in the
first step, which is bounded by~\eqref{equ:10},
and $\rho$ is the total size of all faces computed in the
second step, which is bounded by~\eqref{equ:20}.
As $m<n^2/2$, $\log m=O(\log n)$. Hence, the runtime of the third step is
$$O\left(\left(\frac{n^{2/3}m^{2/3}}{r^2} + \frac{n}{r} +
\frac{n}{r^2}\alpha(\frac{n}{r^2})+\frac{n}{r^2}\log\frac{m}{r^2}+\frac{m}{r^2} \right)\log
n\right).$$
\end{enumerate}

Since $\Xi$ has $O(r^2)$ cells, the total time of the first and third
steps for all cells of $\Xi$ is $O(n^{2/3}m^{2/3}\log n+nr\log
n+m\log n+n\alpha(n/r^2)\log n+n\log(m/r^2)\log n)$.

In summary, the runtime of the overall algorithm excluding the recursive calls is
$$O\left(nr^2\alpha^2(n)\log n+n^{2/3}m^{2/3}\log n+nr\log n+m\log
n+n\log(m/r^2)\log n\right).$$

Let $T(n,m)$ be the total time of the overall algorithm. If $m=1$,
we apply the algorithm for computing a single
face~\cite{ref:AmatoCo95}, and thus $T(n,m)=O(n\alpha^2(n)\log n)$. If
$m\geq n^2/2$, we use the straightforward approach and thus
$T(n,m)=O(m\log n)$. Since $r$ is a constant, we obtain the
following (with big-O notation omitted)
\[
T(n,m)=
\begin{cases}
n\alpha^2(n)\log n & m=1,\\
n^{2/3}m^{2/3}\log n+(n\alpha^2(n)+n\log m+m)\log n + r^2\cdot T(\frac{n}{r^2},\frac{m}{r^2}) & 2\leq m<n^2/2,\\
m\log n & m\geq n^2/2.
\end{cases}
\]
Note that after at most $O(\log m)$ recursions, we will reach
subproblems $T(n,m)$ with $m=1$, and after at most $O(\log
(n/\sqrt{m}))$ recursions, we will reach subproblems $T(n,m)$ with
$m\geq n^2/2$. Hence,
the depth of the recursion is $O(\min\{\log m,\log (n/\sqrt{m})\})$.
Therefore, the recurrence relation solves to $T(n,m)=O(n^{2/3}m^{2/3}\log n+\tau(n\alpha^2(n)+n\log m+m)\log n)$, where $\tau=\min\{\log m,\log (n/\sqrt{m})\}$.

The following theorem summarizes the result.

\begin{theorem}\label{theo:30}
Given a set $S$ of $n$ line segments and a set $P$ of $m$ points in the plane,
the faces of the arrangement of the segments containing at least one point of $P$ can be computed in 
$O(n^{2/3}m^{2/3}\log n+\tau(n\alpha^2(n)+n\log m+m)\log n)$ time, where $\tau=\min\{\log m,\log (n/\sqrt{m})\}$.
\end{theorem}

\paragraph{Remark.}
The algorithm runs in $O(n\alpha^2(n)\log n)$ time for $m= O(1)$,
which matches the time for computing a single
face~\cite{ref:AmatoCo95}, and runs in $O(m\log n)$ time for $m=
\Omega(n^2)$, which matches the performance of the straightforward
approach.

\subsection{The randomized algorithm}

In this section, we present a randomized algorithm, whose running time is a function of $K$, the number of intersections of all segments of $S$. The algorithm is faster than Theorem~\ref{theo:30} when $K=o(n^2)$

We again assume that $m< n^2/2$ since otherwise the problem can be solved in $O(m\log n)$ by the straightforward approach. We resort to a result of de Berg and Schwarzkopf~\cite{ref:deBergCu95}. Given any $r\leq n$ and $K$, de Berg and Schwarzkopf~\cite{ref:deBergCu95} gave a randomized algorithm that can construct a $(1/r)$-cutting $\Xi$ for $S$ in $O(n\log r+Kr/n)$ expected time and the size of $\Xi$ is $O(r+Kr^2/n^2)$. For each cell $\sigma\in \Xi$ (which is a triangle\footnote{In the algorithm description~\cite{ref:deBergCu95}, each cell of the cutting is a constant-sized convex polygon, but we can further triangulate it without increasing the complexity asymptotically.}), $\sigma$ is intersected by $O(n/r)$ segments of $S$.

We set $r=n^2/(n+K)$, and thus $1<r\leq n$ and the size of $\Xi$ is bounded by $O(r)$. By building a point location data structure on $\Xi$~\cite{ref:EdelsbrunnerOp86,ref:KirkpatrickOp83}, we find, for each point of $P$, the cell of $\Xi$ containing it. This takes $O(r+m\log r)$ time in total. For each cell $\sigma\in \Xi$, define $P(\sigma)=P\cap \sigma$. If $|P(\sigma)|>m/r$, then in the same way as in Section~\ref{sec:first}, we further triangulate $\sigma$ into $2\cdot \lceil|P(\sigma)|\cdot r/m\rceil$ triangles each of which contains at most $m/r$ points of $P$; we now consider these triangles as cells of $\Xi$ but $\sigma$ is not a cell of $\Xi$ anymore. As before, if we presort $P$ in $O(n\log n)$ time, then the triangulation for all cells of $\Xi$ can be done in $O(n)$ time in total.

The high-level scheme of the algorithm is similar in spirit to that in Section~\ref{sec:second} for the line case.
For each cell $\sigma\in \Xi$, let $S_{\sigma}$ denote the subset of segments of $S$ intersecting $\sigma$. We define the {\em zone} of $\sigma$ as the collection of face portions of $F\cap \sigma$ for all faces $F\in \calA(S_{\sigma})$ that intersect the boundary of $\sigma$. As in the line case in Section~\ref{sec:second}, to compute all nonempty faces of $\calA(S)$, it suffices to compute, for every cell $\sigma\in \Xi$, the faces of $\calA(S_{\sigma})$ containing the points of $P(\sigma)$ and the zone of $\sigma$. The nonempty faces of $\calA(S)$ that are split among the
zones can be obtained by merging the zones along the edges of cells
of $\Xi$.

Computing the zone for $\sigma$ can be done in $O(n_{\sigma}\alpha(n_{\sigma
})\log n_{\sigma})$ randomized time~\cite{ref:ChazelleCo93}, where $n_{\sigma}=|S_{\sigma}|$.
Then, for each point $p\in P(\sigma)$, we determine whether $p$ belongs to a face of the zone of $\sigma$. This can be done by first building a point location data structure on the zone~\cite{ref:EdelsbrunnerOp86,ref:KirkpatrickOp83} and then using point location queries. As the size of the zone is $O(n_{\sigma}\alpha(n_{\sigma}))$, this step takes $O(n_{\sigma}\alpha(n_{\sigma})+m_{\sigma}\log n_{\sigma})$ time, where $m_{\sigma}=|P(\sigma)|$. Finally, we merge all zones for all cells of $\Xi$ (which is a straightforward task as zones of different cells are disjoint).
As $n_{\sigma}=O(n/r)$, $m_{\sigma}\leq m/r$, and $\Xi$ has $O(r)$ cells, the total time spent on zones is $O(n\alpha(n/r)\log(n/r)+m\log(n/r))$.

For the subproblem of computing the faces of $\calA(S_{\sigma})$ containing the points of $P(\sigma)$, we apply Theorem~\ref{theo:30}. Since we are satisfied with a randomized procedure, in the algorithm of Theorem~\ref{theo:30}, we can replace the $O(n\alpha^2(n)\log n)$ time deterministic algorithm~\cite{ref:AmatoCo95} for computing the zone of a cell edge by a slightly faster $O(n\alpha(n)\log n)$ time randomized algorithm~\cite{ref:ChazelleCo93}, and thus the total expected time of the algorithm for Theorem~\ref{theo:30} is the same as before except that the factor $\alpha^2(n)$ becomes $\alpha(n)$. Applying this result, we can solve the subproblem for $\sigma$ in $O(n_{\sigma}^{2/3}m_{\sigma}^{2/3}\log n_{\sigma}+\tau_{\sigma}(n_{\sigma}\alpha(n_{\sigma})+n_{\sigma}\log m_{\sigma}+m_{\sigma})\log n_{\sigma})$ time, where $\tau_{\sigma}=\min\{\log m_{\sigma},\log(n_{\sigma}/\sqrt{m_{\sigma}}\}$.
As $n_{\sigma}=O(n/r)$, $m_{\sigma}\leq m/r$, and $\Xi$ has $O(r)$ cells, the overall time for solving the subproblem for all cells of $\sigma\in \Xi$ is
\begin{equation*}
O\left(\frac{n^{2/3}m^{2/3}}{r^{1/3}}\log\frac{n}{r}+\tau'\left(n\alpha(\frac{n}{r})+n\log\frac{m}{r}+m\right)\log\frac{n}{r}\right),
\end{equation*}
where $\tau'=\min\{\log (m/r),\log(n/\sqrt{mr})\}$.

Since $r>1$, we have $n/r<n$, $m/r<m$, and $\tau'<\tau$, where $\tau=\min\{\log m,\log (n/\sqrt{m})\}$.
Plugging in $r=n^2/(n+K)$ leads to the total time of the entire algorithm bounded by $O(m^{2/3}K^{1/3}\log n+\tau(n\alpha(n)+n\log m+m)\log n)$.

The above algorithm works based on the assumption that $K$ is known. As we do not know $K$, we overcome the problem by the standard trick of ``guessing''. We start with $K'=K_0$ for a constant $K_0$, and run the algorithm with $K'$. If the algorithm takes too long, then our guess is too low and we double $K'$. Using this doubling strategy, the algorithm is expected to stop within a constant number of rounds after $K'$ is larger than $K$ for the first time. Hence, the total time is asymptotically the same as if we had plugged in the right value of $K$, except for the $O(\tau(n\alpha(n)+n\log m+m)\log n)$ overhead term, which increases by a factor of $O(\log K)$.

We conclude this section with the following theorem.

\begin{theorem}\label{theo:40}
Given a set $S$ of $n$ line segments and a set $P$ of $m$ points in the plane,
the faces of the arrangement of the segments containing at least one point of $P$ can be computed in $O(m^{2/3}K^{1/3}\log n+\tau(n\alpha(n)+n\log m+m)\log n\log K)$ expected time, where $\tau=\min\{\log m,\log (n/\sqrt{m})\}$ and $K$ is the number of intersections of all segments of $S$.
\end{theorem}

\section{The face query problem}
\label{sec:query}

In this section, we consider the face query problem. Let $S$ be a set
of lines in the plane. The problem is to build a data structure on $S$
so that given a query point $p$, the face $F_p(S)$ of the arrangement $\calA(S)$ that contains $p$ can be computed
efficiently. Since $F_p(S)$ is convex, our query algorithm will return
the root of a binary search tree storing $F_p(S)$ so that
binary-search-based queries on $F_p(S)$ can be performed in $O(\log
n)$ time each (e.g., given a query point $q$, decide whether $q\in
F_p(S)$; given a line $\ell$, compute its intersection with $F_p(S)$).
$F_p(S)$ can be output explicitly in $O(|F_p(S)|)$
additional time using the tree.

We work in the dual plane as in Section~\ref{sec:first} and also follow
the notation there. Let $S^*$ denote the set of dual points of $S$.
For a query point $p$ in the primal plane, let $p^*$ denote its dual
line. Define $S^*_+(p^*)$, $S^*_-(p^*)$, $H_+(p^*)$, $H_-(p^*)$, and
$F^*_p(S)$ in the same way as in Section~\ref{sec:first}.

Inspired by the algorithm of Lemma~\ref{lem:40}, we resort to the
randomized optimal partition tree of Chan~\cite{ref:ChanOp12}, which is
originally for simplex range counting queries in
$d$-dimensional space for any constant $d\geq 2$. We briefly review
the partition tree in the planar case. Let $P$ be a set of $n$ points in the
plane. Chan's partition tree recursively subdivides the plane into
triangles (also referred to as {\em cells}). Each node $v$ of $T$
corresponds to a triangle $\triangle_v$ and a subset $P_v$ of $P$ such
that $P_v=P\cap \triangle_v$. If $v$ is an internal node, then $v$ has
$O(1)$ children whose triangles form a disjoint partition of
$\triangle_v$. Hence, each point of $P$ appears in $P_v$ for only one node $v$ in each level of $T$.
$\triangle_v$ and the cardinality $|P_v|$ are stored at $v$.
But $P_{v}$ is not explicitly stored at $v$ unless $v$ is a leaf, in which case $|P_{v}|=O(1)$.\footnote{To simplify the discussion for solving our problem, if $v$ is a leaf and $|P_v|>1$, then we further triangulate $\triangle_v$ into $O(1)$ triangles each of which contains at most one point of $P$. This adds one more level to $T$ but has the property that each leaf triangle contains at most one point of $P$. This change does not affect the performance of the tree asymptotically.}
The space of $T$ is $O(n)$ and its height is $O(\log n)$. $T$ can be built by a randomized algorithm of $O(n\log n)$ expected time. Given a query half-plane $h$, the range query algorithm~\cite{ref:ChanOp12} finds a set $V_h$ of nodes of $T$ such that $P\cap h$ is exactly the union of $P_v$ for all nodes $v\in V_h$, the triangles $\triangle_v$ for all $v\in V_h$ are pairwise disjoint (and thus the subsets $P_v$ for all $v\in V_h$ are also pairwise disjoint), and $|V_h|=O(\sqrt{n})$ holds with high probability. The query algorithm runs in $O(\sqrt{n})$ time with high probability.

\paragraph{Preprocessing.}
To solve our problem, in the preprocessing, we build Chan's partition tree $T$ on the points of $S^*$, which takes $O(n)$ space and $O(n\log n)$ expected time. Let $S^*_v=S^*\cap \triangle_v$ for each node $v\in T$.
We further enhance $T$ as follows. For each node $v\in T$, we compute
the convex hull $H_v$ of $S^*_v$ and store $H_v$ at $v$ by a binary
search tree. To this end, we can presort all points of $S^*$ by
$x$-coordinate. Then, we sort $S^*_v$ for all nodes $v\in T$, which can be done in $O(n\log n)$
time in total due to the presorting of $S^*$.
Consequently, computing the convex hull $H_v$ can be done in
$O(|S^*_v|)$ time. As such, computing convex hulls
for all nodes of $T$ takes $O(n\log n)$ time in total. With these
convex hulls, the space of $T$ increases to $O(n\log n)$, because the height of $T$ is $O(\log n)$ and the subsets $S^*_v$ for all nodes $v$ in the same level of $T$ form a partition of $S^*$. This finishes our
preprocessing, which takes $O(n\log n)$ space and $O(n\log n)$
expected time.

\paragraph{Queries.}
Consider a query point $p$. Without loss of generality, we assume that
$p^*$ is horizontal. Using the partition tree $T$, we compute the lower hull $H_+(p^*)$
as follows.

Let $h$ be the upper half-plane bounded by the line $p^*$.
We apply the range query algorithm~\cite{ref:ChanOp12} on $h$ and find a
set $V_h$ of nodes of $T$, as discussed above. According to the
properties of $V_h$, $S^*_+(p^*)$ is the union of $S^*_v$ for all
nodes $v\in V_h$ and the triangles $\triangle_v$ for all $v\in V_h$
are pairwise disjoint. Therefore, $H_+(p^*)$ is the lower hull of the
convex hulls $H_v$ of all $v\in V_h$. As the triangles of
$\triangle_v$ for all $v\in V_h$ are pairwise disjoint, the convex
hulls $H_v$ of all $v\in V_h$ are also pairwise disjoint.
Consequently, we can apply the algorithm of Lemma~\ref{lem:40} to
compute $H_+(p^*)$ from convex hulls $H_v$ of all $v\in V_h$, which
takes $O(|V_h|\log n)$ time because convex hulls $H_v$ are already
available due to the preprocessing. As $|V_h|=O(\sqrt{n})$ holds with high
probability, the time for computing $H_+(p^*)$ is bounded by
$O(\sqrt{n}\log n)$ with high probability.

Analogously, we can compute the upper hull $H_-(p^*)$. Afterwards, $F_p(S)$ can be
obtained as a binary search tree in $O(\log n)$ time by
computing the inner common tangents of $H_+(p^*)$ and $H_-(p^*)$, as explained in Section~\ref{sec:first}. The total query time is bounded by $O(\sqrt{n}\log n)$ with high probability. Further, $F_p(S)$ can be output explicitly in additional $|F_p(S)|$ time.

As discussed in Section~\ref{sec:first}, once the binary search tree
for $F_p(S)$ is constructed, binary search trees representing
	convex hulls of some nodes of $T$ may be destroyed unless fully
	persistent trees are used. To handle future queries, we need to
	restore those convex hulls.
Different from the algorithm in Section~\ref{sec:first}, depending on
applications, persistent trees may not be necessary here. For example,
if $F_p(S)$ needs to be output explicitly, then after $F_p(S)$ is
	output, we can restore those destroyed convex hulls by ``reversing'' the operations that are performed
	during the algorithm of Lemma~\ref{lem:40}. The time is still
	bounded by $O(\sqrt{n}\log n)$ with high probability. Hence in
	this case persistent trees are not necessary. Also, if $F_p(S)$
	only needs to be implicitly represented but $F_p(S)$ will not be needed
	anymore before the next query is performed, then we can also
	restore the convex hulls as above without
	using persistent trees. However, if $F_p(S)$ only needs to be
	implicitly represented and $F_p(S)$ still needs to be kept even
	after the next query is performed, then we have to use persistent
	trees.

\medskip

We summarize our result in the following theorem.

\begin{theorem}\label{theo:query10}
Given a set $S$ of $n$ lines in the plane, we can preprocess it in
$O(n\log n)$ randomized time and $O(n\log n)$ space so that for any
query point $p$, we can produce a binary search tree representing the face of $\calA(S)$
that contains $p$ and the query time is bounded by $O(\sqrt{n}\log n)$ with high
probability. Using the binary search tree,
standard binary-search-based queries on the face can be performed in
$O(\log n)$ time each,
and outputting the face explicitly can be done in additional time linear in the
number of edges of the face.
\end{theorem}

As discussed in Section~\ref{sec:intro}, using our result in Theorem~\ref{theo:query10}, the algorithm in~\cite{ref:EdelsbrunnerIm89} for the face query problem in the segment case can also be improved accordingly.


\subsection{Tradeoff between storage and query time}

We further obtain a tradeoff between the preprocessing and the query
time. To this end, we make use of Chan's $r$-partial partition
tree~\cite{ref:ChanOp12}. Let $P$ be a set of $n$ point in the plane.
For any value $r<n/\log^{\omega(1)} n$,
an {\em $r$-partial partition tree} $T(r)$ for $P$ is the same as a
partition tree discussed before, except that a leaf now may contain up
to $r$ points. The number of nodes of $T$ is $O(n/r)$.
$T(r)$ can be built in $O(n\log n)$ randomized time.
Given a query half-plane $h$, the range query
algorithm~\cite{ref:ChanOp12} finds two sets $V^1_h$ and $V_h^2$ of
nodes of $T(r)$ with the following property: (1) for each node $v\in
V^1_h$, the triangle $\triangle_v$ is inside $h$; (2) for each node
$v\in V^2_h$, $v$ is a leaf and $\triangle_v$ is crossed by the bounding line of $h$;
(3) $P\cap h$ is the union of $P_v$ for all nodes $v\in V_h^1$ as well as
the intersection $P_v\cap h$ for all nodes $v\in V_h^2$; (4) the triangles $\triangle_v$ for all $v\in V_1(h)\cup V_2(h)$ are pairwise disjoint; (5)
$|V_1(h)|+|V_2(h)|=O(\sqrt{n/r})$ holds with high
probability. The query algorithm finds $V_1(h)$ and
$V_2(h)$ in $O(\sqrt{n/r})$ time with high probability.

\paragraph{Preprocessing.}
To solve our problem, in the preprocessing we build an $r$-partial
partition tree $T(r)$ on the points of $S^*$. For each node $v\in
T(r)$, we still compute and store the convex hull $H_v$ of $P_v$. This
still takes $O(n\log n)$ space and $O(n\log n)$ expected time as
before. Next, we perform additional preprocessing for each leaf $v$ of $T(r)$. Note
that $|S^*_v|\leq r$. Let $S_v$ denote the subset of the lines of $S$
in the primal plane dual to the points of $S^*_v$. We compute
explicitly the arrangement $\calA(S_v)$. For each face $F\in
\calA(S_v)$, its leftmost and rightmost vertices divide the boundary
of $F$ into an upper portion and a lower portion; for each portion, we use a binary
search tree to store it. We also build a point location data structure
on $\calA(S_v)$~\cite{ref:EdelsbrunnerOp86,ref:KirkpatrickOp83}.
This finishes the preprocessing for $v$, which takes $O(r^2)$ time and space. As
$T(r)$ has $O(n/r)$ leaves, the preprocessing for all leaves takes
$O(nr)$ time and space. Overall, the preprocessing
takes $O(n\log n+nr)$ expected time and $O(n\log
n+nr)$ space.

\paragraph{Queries.}
Consider a query point $p$. Again, we assume that
its dual line $p^*$ is horizontal. We compute the lower hull $H_+(p^*)$
as follows.

Let $h$ be the upper half-plane bounded by the line $p^*$.
We apply the range query algorithm~\cite{ref:ChanOp12} on $h$
and find two sets $V^1_h$ and $V^2_h$ of nodes of $T(r)$, as discussed
above. Due to the property (3) of $V^1_h$ and $V^2_h$ discussed
above,
$H_+(p^*)$ is the lower hull of the convex hulls $H_v$ of all $v\in V^1_h$ and the convex hulls $H'_v$ of the subset of points of $S^*_v$ above the line $p^*$ for all $v\in V^2_h$.
For each $v\in V^1_h$, the convex hull $H_v$ is available due to the preprocessing.
For each $v\in V^2_h$, $H'_v$ can be obtained in $O(\log n)$ time as
follows. Using the point location data structure on $\calA(S_v)$,
we find the face $F_p(S_v)$ of $\calA(S_v)$ containing $p$, and then
$H'_v$ is dual to the lower portion of the boundary of $F_p(S_v)$\footnote{In
fact, the lower portion of the boundary of $F_p(S_v)$ may be only
part of the dual of $H'_v$. However, since $F_p(S)\subseteq
F_p(S_v)$, using the lower portion of $F_p(S_v)$ as the dual of $H'_v$ to
compute $H_+(p^*)$ and then compute $F_p(S)$ will give the correct
answer.}, whose binary search tree is computed in the preprocessing.
Due to the property (4) of $V^1_h$ and $V^2_h$, all convex hulls
$H_v$, $v\in V^1_h$, and $H_v'$, $v\in V^2_h$, are pairwise disjoint.
Thus, we can again apply the algorithm of Lemma~\ref{lem:40} to compute
$H_+(p^*)$ from these convex hulls in $O((|V^1_h|+|V^2_h|)\log n)$
time. As $|V^1_h|+|V^2_h|=O(\sqrt{n/r})$ holds with high probability,
the time for computing $H_+(p^*)$ is bounded by $O(\sqrt{n/r}\log n)$
with high probability.


Analogously, we can compute upper hull $H_-(p^*)$. Afterwards, $F_p(S)$ can be
obtained as a binary search tree in $O(\log n)$ time by
computing the inner common tangents of $H_+(p^*)$ and $H_-(p^*)$, as
explained in Section~\ref{sec:first}. The total query time is bounded
by $O(\sqrt{n/r}\log n)$ with high probability. Further, $F_p(S)$ can be
output explicitly in additional $|F_p(S)|$ time.

As before, depending on applications, one can decide whether
persistent trees are needed for representing convex hulls of the nodes
of $T(r)$ as well as the boundary portions of the faces
of the arrangements $\calA(S_v)$ of the
leaves $v$ of $T(r)$.

\medskip

We summarize our result in the following
theorem.

\begin{theorem}\label{theo:query20}
Given a set $S$ of $n$ lines in the plane, for any value
$r<n/\log^{\omega(1)} n$, we can preprocess it in $O(n\log n+nr)$
randomized time and $O(n\log n+nr)$ space so that for any query point
$p$, we can produce a binary search tree representing the face of $\calA(S)$
that contains $p$ and the query time is bounded by $O(\sqrt{n/r}\log n)$ with high
probability. Using the binary search tree,
standard binary-search-based queries can be performed on the face in
$O(\log n)$ time each, and outputting the face explicitly can be done in time linear in the
number of edges of the face.
\end{theorem}

\paragraph{Remark.}
Using the random sampling
techniques~\cite{ref:ClarksonNe87,ref:HausslerEp87}, a tradeoff
between the preprocessing and the query time was also provided
in~\cite{ref:EdelsbrunnerIm89} roughly with the following performance:
the preprocessing takes $O(n^{3/2}r^{1/2}\log^{3/2}r\log^2 n)$ randomized time, the space
is $O(nr\log r\log n)$, and the query time is bounded by
$O(\sqrt{n/r}\log^2 n)$ with high probability (combining with the
compact interval trees~\cite{ref:GuibasCo91}). Hence, our result improves on all three aspects, albeit on a smaller range of $r$.



\footnotesize
 \bibliographystyle{plain}
\bibliography{reference}

\end{document}